\documentclass[lettersize,journal]{IEEEtran}
\usepackage{cite}
\usepackage{amsmath,amssymb,amsfonts}
\usepackage{amsthm}
\usepackage{algorithmic}
\usepackage{graphicx}
\usepackage{textcomp}
\usepackage{algorithm}
\usepackage{array}
\usepackage[caption=false,font=normalsize,labelfont=sf,textfont=sf]{subfig}
\usepackage{stfloats}
\usepackage{booktabs}
\usepackage[hyphens]{url}
\usepackage{verbatim}
\usepackage{soul}
\usepackage{balance}
\usepackage{stmaryrd} 
\usepackage[colorlinks,
linkcolor=black,
anchorcolor=black,
citecolor=black
]{hyperref}

\newtheorem{definition}{Definition}
\newtheorem{theorem}{Theorem}

\newcommand{\main}{ObliuSky}
\newcommand{\csa}{$CS_{\{1,2\}}$}
\newcommand{\llb}{\llbracket}
\newcommand{\rrb}{\rrbracket}

\begin{document}

\title{ObliuSky: Oblivious User-Defined Skyline Query Processing in the Cloud}

\author{Yifeng Zheng, Weibo Wang, Songlei Wang, Zhongyun Hua, and Yansong Gao
    \thanks{Yifeng Zheng, Weibo Wang, Songlei Wang, and Zhongyun Hua are with the School of Computer Science and Technology, Harbin Institute of Technology, Shenzhen, Guangdong 518055, China (e-mail: yifeng.zheng@hit.edu.cn, weibo.wang.hitsz@outlook.com, songlei.wang@outlook.com, huazhongyun@hit.edu.cn).}
    
    \thanks{Yansong Gao is with Data61, CSIRO, Sydney, Australia (e-mail: garrison.gao@data61.csiro.au).}
}

\IEEEtitleabstractindextext{
		\begin{abstract}

        The proliferation of cloud computing has greatly spurred the popularity of outsourced database storage and management, in which the cloud holding outsourced databases can process database queries on demand. Among others, skyline queries play an important role in the database field due to its prominent usefulness in multi-criteria decision support systems. To accommodate the tailored needs of users, user-defined skyline query has recently emerged as an intriguing advanced type of skyline query, which allows users to define custom preferences in their skyline queries (including the target attributes, preferred dominance relations, and range constraints on the target attributes). However, user-defined skyline query services, if deployed in the cloud, may raise critical privacy concerns as the outsourced databases and skyline queries may contain proprietary/privacy-sensitive information, and the cloud might even suffer from data breaches. In light of the above, this paper presents ObliuSky, a new system framework enabling oblivious user-defined skyline query processing in the cloud. ObliuSky departs from the state-of-the-art prior work by not only providing confidentiality protection for the content of the outsourced database, the user-defined skyline query, and the query results, but also making the cloud oblivious to the data patterns (e.g., user-defined dominance relations among database points and search access patterns) which may indirectly cause data leakages. We formally analyze the security guarantees and conduct extensive performance evaluations. The results show that while achieving much stronger security guarantees than the state-of-the-art prior work, ObliuSky is superior in database and query encryption efficiency, with practically affordable query latency.
      \end{abstract}
		
		\begin{IEEEkeywords}
			Database service outsourcing, secure user-defined skyline query, cloud computing, privacy preservation
		\end{IEEEkeywords}
	}

\maketitle

\IEEEdisplaynontitleabstractindextext
	
\IEEEpeerreviewmaketitle

\section{Introduction}

It has been increasingly common for enterprises and organizations to outsource the storage and management of databases to the cloud, which can then provide services for processing queries over the outsourced databases. 
Among others, skyline queries play an important role in the database field, garnering wide attentions in various application domains due to its prominent advantages in multi-criteria decision support systems \cite{balke2004efficient, huang2006skyline, deng2007multi}.
However, skyline query services, if deployed in the cloud, may raise critical privacy concerns regarding the outsourced databases and skyline queries, which may contain proprietary/privacy-sensitive information.
Moreover, the cloud might even suffer from data breaches \cite{QinW0018, JiangWHWLSR21} which would seriously harm data privacy.
These critical concerns necessitate the development of security mechanisms for deploying the skyline query service in the cloud, providing protections for the outsourced databases, skyline queries, and query results.

With traditional skyline queries, the user specifies a query point/tuple with same number of dimensions as database points/tuples, and retrieves a set of skyline points that are not dominated by any other point in the database \cite{liu2019secure}.
In some real-world scenarios, however, users may be only interested in certain dimensions (each dimension corresponds to an attribute) and prefer maximum or minimum value on each target dimension.
Besides, users may want to specify a constrained region to filter out database points that do not meet their expectations.
Driven by these practical requirements, user-defined skyline query \cite{zhang2022efficient} has recently emerged as an intriguing type of skyline query to accommodate tailored needs of users. It allows users to customize their target dimensions, corresponding preferences, and constrained regions.
%
%

Consider for example an application that allows users to choose a subset of hotels from a number of hotels around a tourist spot, based on skyline queries.
Each hotel is assumed to be associated with three attributes: \emph{price}, \emph{distance} (between the hotel and the tourist spot), and \emph{score} (which indicates the hotel's service quality).
In practice, some users may care more about the price and distance attributes, while others may care little about the price and are more concerned with the score and distance attributes.
Hence, different users may select different attributes in forming their skyline queries.
Meanwhile, for the same attribute, different users may have different preferences on its value.
For example, some users may prefer to stay as close as possible to the tourist spot, so they prefer minimal value on the distance attribute.
Other users may prefer not to stay near the tourist spot which is probably noisy, so they prefer maximum value on the distance attribute.
Moreover, for each target attribute, a user may want to specify a range constraint that its value should satisfy.
These range constraints form a constrained region that the retrieved skyline points should fall within.
For example, different users may have different budgets for the hotel payment, so it is natural for them to specify different range constraints on the price attribute.

In applications like the above kind, user-defined skyline query, which allows users to select attributes/dimensions, specify preferences, and define constrained regions, is much better suited for flexibly meeting tailored user needs.
So in this paper our focus is on exploring how to enable oblivious user-defined skyline query processing in the cloud.
%
Specifically, given an encrypted outsourced database and an encrypted user-defined skyline query, we aim to allow the cloud to obliviously fetch the encrypted skyline points/tuples.
Here, for securing the service, we not only aim to provide confidentiality protection for the content of the outsourced database, the user-defined skyline query, and the query results, but also ambitiously aim to make the cloud oblivious to the data patterns which may indirectly cause data leakages \cite{liu2019secure,ding2021efficient}.
%
Specifically, the data patterns include the dominance relationships among the database points, the number of database points that each skyline point dominates, as well as the search access patterns.
Here the search pattern reveals whether a new user-defined skyline query has been issued previously and the access pattern reveals which points in the database are output as the skyline points.

Skyline query with privacy preservation has received growing attentions in recent years.
Most existing works are focused on traditional skyline query \cite{liu2017secure,liu2019secure,ding2021efficient,zheng2022secskyline,wang2020stargazing}.
Little work \cite{liu2018pusc,zhang2022efficient} has been done on the privacy-preserving user-defined skyline query.
In the state-of-the-art prior work \cite{zhang2022efficient}, Zhang \emph{et al.} present the formal definition of user-defined skyline query and propose a scheme for privacy-preserving user-defined skyline query.
However, the practical usability of their scheme is hindered by prominent security downsides, as it leaks the aforementioned data patterns (i.e., the dominance relationships among the database points, the number of database points that each skyline point dominates, and the search access patterns) to the cloud.
Therefore, how to enable oblivious user-defined skyline query processing that provides protection for both data confidentiality and data patterns is still challenging and remains to be fully explored.

In light of the above, in this paper we propose {\main}, a new system framework enabling oblivious user-defined skyline query processing in the cloud.
{\main} builds on lightweight secret sharing technique \cite{demmler2015aby} for efficient data encryption, and allows the cloud to hold an encrypted database and obliviously process encrypted user-defined skyline query. 
%
In comparison with the state-of-the-art prior work \cite{zhang2022efficient}, {\main} works in the encrypted domain while safeguarding the aforementioned data patterns.
Our key insight is to let the cloud first obliviously shuffle the tuples in the encrypted original database and then work over the encrypted shuffled database to securely find encrypted skyline tuples. 
%
%
In this way, the dominance relationships among the (encrypted) tuples in the (encrypted) \emph{shuffled} database can be safely revealed to identify which tuples in the \emph{shuffled} database are skyline tuples, as the shuffle is performed in an oblivious manner and the permutation is unknown to the cloud.
In short, based on the key idea of oblivious shuffle, {\main} only reveals \emph{untrue} data patterns to the cloud while allowing it to correctly and obliviously locate encrypted skyline tuples.

The security guarantees of {\main} are formally analyzed. We implement {\main}'s protocols and conduct extensive experiments for performance evaluation.
The results show that while achieving much stronger security guarantees than the state-of-the-art prior work \cite{zhang2022efficient}, {\main} is also superior in database and query encryption efficiency. Additionally, as a natural trade-off for strong security guarantees, the query latency is higher than that in \cite{zhang2022efficient}, but is still practically affordable (only about 2.4 seconds on an encrypted database of 10000 5-dimensional tuples, with 3 user-selected dimensions and the percentage of database tuples satisfying user-defined constrained region being 1\%).
Our main contributions are summarized below:

\begin{itemize}
    \item We present {\main}, a new system framework enabling oblivious user-defined skyline query processing in the cloud and achieving much stronger security guarantees than the state-of-the-art prior work.

    \item We propose the key idea of obliviously shuffling the encrypted outsourced database and revealing untrue data patterns in subsequent processing for each encrypted user-defined skyline query. 


   \item We design an oblivious user-defined dominance evaluation protocol allowing the cloud to obliviously evaluate the user-defined dominance relationships between database tuples. Based on this protocol, we further develop a protocol for oblivious fetch of encrypted skyline tuples, while keeping the data patterns private.

    \item We formally analyze the security of {\main} and conduct comprehensive performance evaluations.
    The results show that while achieving much stronger security guarantees than the state-of-the-art prior work \cite{zhang2022efficient}, ObliuSky is also better in database and query encryption efficiency, with practically affordable query latency.

\end{itemize}

The rest of this paper is organized as follows.
Section \ref{sec:Related} discusses the related work.
In Section \ref{sec:Preliminaries}, we introduce the preliminaries.
Then, we introduce our system architecture, threat assumptions, and security guarantees in Section \ref{sec:System}.
After that, we present the design of {\main} in Section \ref{sec:Design}, followed by security analysis and experiments in Section \ref{sec:Security} and \ref{sec:Experiments}.
Finally, we conclude this paper in Section \ref{sec:Conclusion}.

\section{Related Work}

\label{sec:Related}

\noindent\textbf{Skyline query processing in plaintext domain.} B\"{o}rzs\"{o}nyi \textit{et al.} \cite{borzsonyi2001skyline} present the first work that studies skyline queries for database systems and builds on the block-nested-loop (BNL) algorithm for processing skyline queries.
After that, skyline query processing has received great attentions in the database field and a number of variants of skyline queries for different application scenarios have been put forward, such as uncertain skyline \cite{liu2015finding,pei2007probabilistic}, subspace skyline \cite{dellis2006constrained,pei2005catching}, and group-based skyline \cite{liu2015findinga,yu2017fast}.
However, this line of work addresses skyline query processing in the plaintext domain without considering privacy protection.

\noindent\textbf{Skyline query processing with privacy awareness.} The first study on skyline query processing with data privacy protection is initiated by Bothe \textit{et al.} \cite{bothe2014skyline}. However, as the first attempt, their proposed scheme lacks formal security guarantees.
%
Later, the works in \cite{liu2017secure,liu2019secure,ding2021efficient} build on additively homomorphic encryption to devise schemes that can provide formal security guarantees under the two-server model. Their schemes provide comprehensive data protections, safeguarding the confidentiality of the outsourced databases, queries, and query results as well as data patterns regarding the dominance relationships between database tuples, the number of database tuples dominated by each skyline tuple, and search access patterns.
Despite the strong security guarantees, the schemes in \cite{liu2017secure,liu2019secure,ding2021efficient} incur high performance overheads due to the use of heavy cryptosystem.
Subsequently, Zheng \textit{et al.} \cite{zheng2022secskyline} present SecSkyline which only makes use of lightweight secret sharing techniques and greatly outperforms the schemes in \cite{liu2017secure,liu2019secure,ding2021efficient} in online query latency.
Different from these cryptographic approaches, Wang \textit{et al.} \cite{wang2020stargazing} propose an alternative approach based on the use of trusted hardware, yet it is known that trusted hardware is subject to various side-channel attacks \cite{hahnel2017high,van2017telling,lee2017inferring,lee2020off}. 
In independent work \cite{wang2022efficient}, Wang \textit{et al.} focus on the support for result verification under specific location-based skyline queries where the skyline query only includes two spatial attributes.
%
%
%
Despite being valuable, all these existing works target application scenarios of skyline queries different from ours and cannot be used for supporting privacy-preserving \emph{user-defined} skyline queries.

The state-of-the-art prior work that is most related to our work is due to Zhang \textit{et al.} \cite{zhang2022efficient}.
They study privacy-aware user-defined skyline query, where user preference is integrated into constrained subspace skyline query \cite{dellis2006constrained}, allowing a user to specify its interested dimensions and define range constraints and dominance measurements for the selected dimensions.
%
%
Their scheme is designed under a single-server model.
Despite being a valuable research endeavor, their scheme provides weak security guarantees.
In particular, in their scheme the dominance relationships between database tuples are disclosed during the protocol execution at the cloud side. 
This further allows the cloud to learn the number of database tuples dominated by each skyline tuple, and the access pattern (which in turn would disclose the search pattern).
%
%
%
%
However, according to prior works on privacy-preserving skyline queries \cite{liu2019secure,ding2021efficient,zheng2022secskyline}, the above leakages should be prevented for strong privacy protection, because they may allow the cloud to indirectly infer private information regarding the databases and queries.
Indeed, hiding the data patterns is what makes the protocol design for privacy-preserving skyline query processing challenging.
In comparison with \cite{zhang2022efficient}, {\main} is free of the above leakages and offers strong security guarantees.

\section{Preliminaries}
\label{sec:Preliminaries}

\subsection{User-Defined Skyline Query}

We follow the state-of-the-art prior work \cite{zhang2022efficient} in the definition of user-defined skyline query.
At a high level, user-defined skyline query is an integration of user preferences with constrained subspace skyline query \cite{dellis2006constrained}.



\begin{definition}
	\label{def:1}
    (User-Defined Dominance). Given a database $\mathcal{T}=\{\mathbf{t}_1,\cdots,\mathbf{t}_n\}$, where each tuple of the database $\mathbf{t}_i$ ($i\in[1,n]$) is an $m$-dimensional vector and each dimension corresponds to an attribute.
    %
    Let $\mathcal{D}=\{1,2,\cdots, m\}$ denote the set of all dimensions and $\mathcal{B}=\{d_1, d_2, \cdots, d_k\}$ denote a user-selected subset of $\mathcal{D}$ with size $k$.
    A tuple $\mathbf{t}_a\in \mathcal{T}$ is said to dominate another tuple $\mathbf{t}_b \in \mathcal{T}$ on the selected dimension subset $\mathcal{B}$ as per user preference, denoted as $\mathbf{t}_a\prec_{\mathcal{B}}\mathbf{t}_b$, if (1) for every $d_j\in\mathcal{B}$, $\mathbf{t}_a[d_j]\leq\mathbf{t}_b[d_j]$ (if the user prefers minimal value on dimension $d_j$) or $\mathbf{t}_a[d_j]\geq\mathbf{t}_b[d_j]$ (if the user prefers maximum value on dimension $d_j$); and (2) for at least one $d_j \in\mathcal{B}$, $\mathbf{t}_a[d_j]\neq \mathbf{t}_b[d_j]$.
\end{definition}

\begin{definition}
    (User-Defined Constrained Region). A user-defined constrained region with respect to a user-selected dimension subset $\mathcal{B}\subseteq\mathcal{D}$ of size $k$ is represented by a set of range constraints $\mathcal{R}=\{(r_{j,l},r_{j,u})\}_{j=1}^k$, where the values $r_{j,l}$ and $r_{j,u}$ represent the user-specified minimum value and maximum value on dimension $d_j$ respectively.
\end{definition}

\begin{definition}
    \label{def:UDSkyline}
    (User-Defined Skyline Query). A user-defined skyline query $Q$ consists of a user-selected dimension subset $\mathcal{B}$, a user-defined constrained region $\mathcal{R}$, and a preference set $\mathcal{P}=\{p_1, p_2, \cdots, p_k\}$ over the dimensions in $\mathcal{B}$ (where $p_j=1$ (resp. $p_j=0$) indicates that the user prefers maximum (resp. minimum) value on dimension $d_j$. The query result is a set of tuples satisfying the range constraints in $\mathcal{R}$ and are not dominated by any other tuple on $\mathcal{B}$ as per the user's preferences.
    Specifically, let $\mathcal{C}=\{\mathbf{t}\in\mathcal{T}|\forall d_j \in\mathcal{B}, {r}_{j,l}\leq\mathbf{t}[d_j]\leq {r}_{j,u}\}$ denote a sub-database of $\mathcal{T}$ that satisfies the user-defined constrained region $\mathcal{R}$. The user-defined skyline query result can be formally denoted by $\mathcal{S}=\{\mathbf{p}\in\mathcal{C}|\nexists \mathbf{p'}\in\mathcal{C}: \mathbf{p'} \prec_{\mathcal{B}}\mathbf{p}\}$.

    
\end{definition}

\subsection{Skyline Computation Based on Block-Nested Loop}
\label{sec:BNL}
Following the state-of-the-art prior work \cite{zhang2022efficient}, we build on the block-nested loop (BNL) algorithm \cite{borzsonyi2001skyline} for skyline computation, which is a widely popular algorithm and supports any dimensionality without requiring indexing or sorting on data. 
The main idea of the BNL algorithm is to maintain a window of candidate skyline tuples (denoted as $\mathcal{X}$) and update $\mathcal{X}$ while scanning the input database.
Specifically, given a database $\mathcal{T}$, in the beginning the first tuple in $\mathcal{T}$ is directly inserted into $\mathcal{X}$.
Then for each subsequent tuple $\mathbf{t}\in\mathcal{T}$, there are three cases to handle:
\begin{itemize}
    \item $\exists\mathbf{x}\in\mathcal{X}$, if $\mathbf{x}\prec\mathbf{t}$, then $\mathbf{t}$ is discarded;
    \item $\exists\mathbf{x}\in\mathcal{X}$, if $\mathbf{t}\prec\mathbf{x}$, then $\mathbf{x}$ is removed from $\mathcal{X}$;
    \item $\mathbf{t}$ is incomparable with all tuples in $\mathcal{X}$, then $\mathbf{t}$ is inserted into $\mathcal{X}$.
\end{itemize}
Here $\prec$ indicates the dominance relationship (the definition of which could be customized in user-defined skyline query), e.g., $\mathbf{x}\prec\mathbf{t}$ means that $\mathbf{x}$ dominates $\mathbf{t}$.

\subsection{Additive Secret Sharing}
\label{sec:secretSharing}

Additive secret sharing (ASS) \cite{demmler2015aby} divides a private value $x$ in the ring $\mathbb{Z}_{2^l}$ into two shares $\langle x\rangle_{1}\in\mathbb{Z}_{2^l}$ and $ \langle x\rangle_{2}\in\mathbb{Z}_{2^l}$, where $x=\langle x\rangle_{1}+ \langle x\rangle_{2}$.
$\langle x\rangle_{1}\in\mathbb{Z}_{2^l}$ and $ \langle x\rangle_{2}\in\mathbb{Z}_{2^l}$ are then distributed to two parties $P_{1}$ and $P_{2}$, respectively.
If $l=1$,  the secret sharing is called \textit{binary sharing} (denoted as $\llbracket x\rrbracket^B$, and its secret shares are denoted as $\langle x\rangle_{1}^B$ and $ \langle x\rangle_{2}^B$), and otherwise \textit{arithmetic sharing} (denoted as $\llbracket x\rrbracket^A$, and its secret shares are denoted as $\langle x\rangle_{1}^A$ and $ \langle x\rangle_{2}^A$).

In the arithmetic ASS domain, the basic operations are as follows (all operations are with respect to the ring $\mathbb{Z}_{2^l},l\geq2$).
(1) To compute $\llbracket c\cdot x \rrbracket^A$ given a secret-shared value $\llbracket x \rrbracket^A$ and a public constant $c$, each party $P_{i},i\in\{1,2\}$ computes $c\cdot\langle x\rangle_{i}^A$ locally. 
(2) To compute $\llbracket x\pm y \rrbracket^A$ given two secret-shared values $\llbracket x \rrbracket^A$ and $\llbracket y \rrbracket^A$, each party $P_{i},i\in\{1,2\}$ computes $\langle x\rangle_{i}^A\pm\langle y\rangle_{i}^A$ locally. 
(3) To compute $\llb x\cdot y\rrb^A$ given two secret-shared values $\llbracket x \rrbracket^A$ and $\llbracket y \rrbracket^A$, $P_{\{1,2\}}$ need to prepare a secret-shared Beaver triple $(\llb w\rrb^A,\llb u\rrb^A,\llb v\rrb^A)$ offline, where $w= u\cdot v$  \cite{RiaziWTS0K18}.
Each $P_{i},i\in\{1,2\}$ first locally computes $\langle e\rangle^A_i=\langle x\rangle^A_i-\langle u\rangle^A_i$ and $\langle f\rangle^A_i=\langle y\rangle^A_i-\langle v\rangle^A_i$. Then, $e$ and $f$ are reconstructed and obtained by the two parties.
Finally, $P_1$ and $P_2$ locally compute the secret shares of $\llb x\cdot y\rrb^A$ by $\langle x\cdot y\rangle^A_1=e\cdot f+f\cdot\langle u\rangle^A_1+e\cdot\langle v\rangle^A_1+\langle w\rangle^A_1$
and $\langle x\cdot y\rangle^A_2=f\cdot\langle u\rangle^A_2+e\cdot\langle v\rangle^A_2+\langle w\rangle^A_2$,
respectively. 
For simplicity, we write $\llb z\rrb^A=\llb x\rrb^A\cdot\llb y\rrb^A$ to represent the secret-shared multiplication.

Note that in the binary ASS domain, addition (+) and subtraction (-) operations are replaced with XOR ($\oplus$) operations, while multiplication ($\cdot$) operations are replaced with AND ($\otimes$) operations.
Additionally, the NOT operation (denoted by $\neg$) in binary ASS domain can be performed by simply having one of $P_1$ and $P_2$ locally flip the secret share it holds. 
For example, to calculate $\llb \neg x\rrb^B$ given $\llb x\rrb^B$, we can simply have $P_1$ flip $\langle x\rangle^B_1$ to produce $\langle \neg x\rangle_{1}^B=\neg\langle x\rangle^B_1$, and the other share $ \langle \neg x\rangle_{2}^B= \langle x\rangle_{2}^B$.

\section{System Overview}
\label{sec:System}

\subsection{Architecture}

\begin{figure}[t!]
    \centering
    \includegraphics[width = \linewidth]{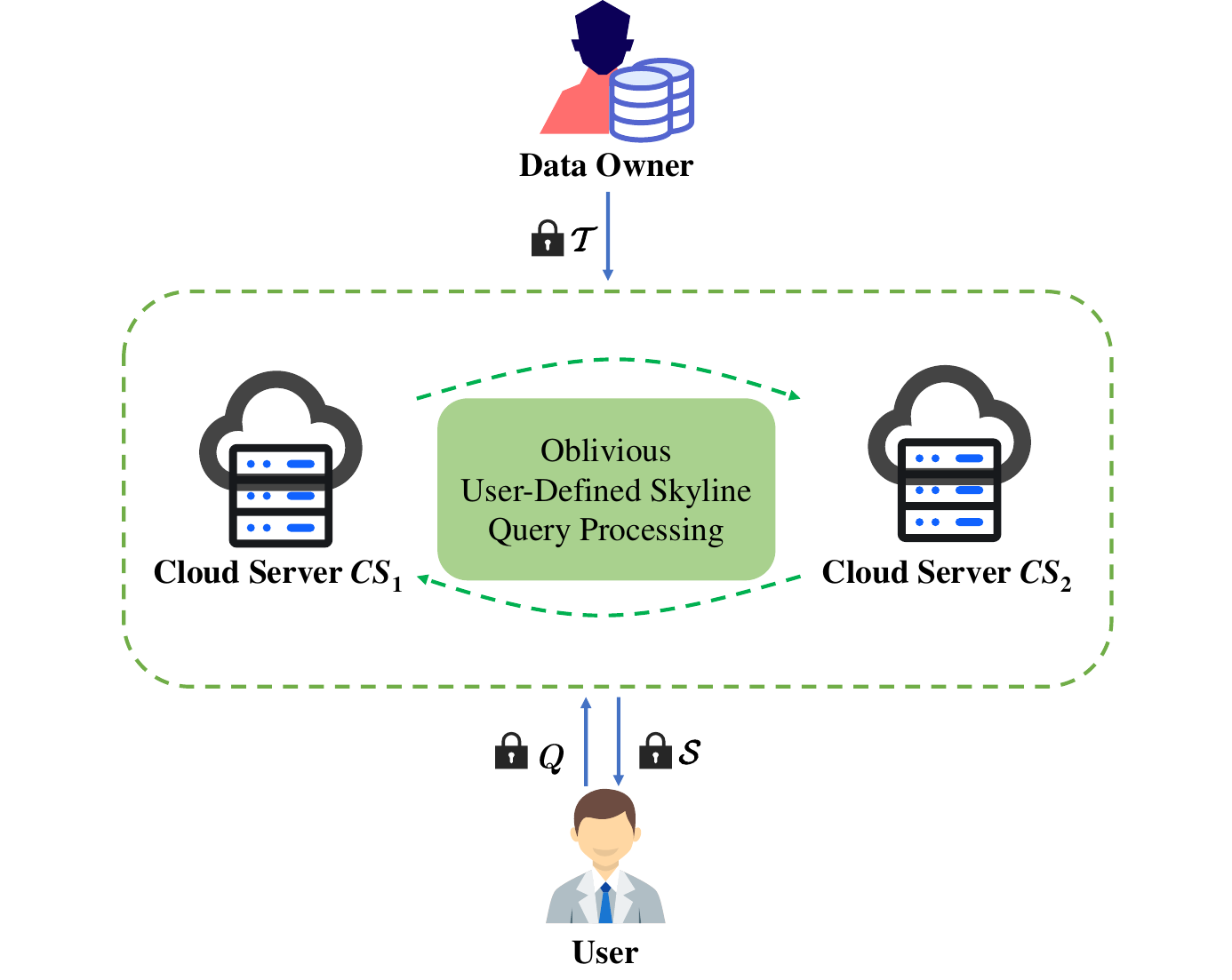}
    \caption{The architecture of {\main}.}
    \label{fig:Arch}
\end{figure}

The system architecture of {\main} is depicted in Fig. \ref{fig:Arch} which is comprised three kinds of entities: the data owner, the user, and the cloud.
The data owner is a trusted entity who owns a database $\mathcal{T}$ of multi-dimensional tuples and wants to provide skyline query service to the users via harnessing the power of cloud computing.
However, as the database is proprietary, the data owner is not willing to directly expose the database in cleartext to the cloud, and instead stores an encrypted version of the database in the cloud.
The cloud can then process user-define skyline queries from the user in a secure manner.
For a user-defined skyline query $Q$, the user needs to specify a dimension subset $\mathcal{B}$, a constrained region $\mathcal{R}$, and a preference set $\mathcal{P}$ over the selected dimension subset.
In {\main}, the user-defined skyline query will also be provided to the cloud in protected form, which then conducts secure search for skyline tuples in the ciphertext domain as per the protocol design of {\main}.
Once the secure search process is completed, the cloud returns a set $\mathcal{S}$ of encrypted skyline tuples to the user for decryption.

In pursuit of performance efficiency, {\main} builds on the lightweight cryptographic technique---additive secret sharing---to protect the outsourced database and the user-defined skyline query, as well as support subsequent secure computation to produce encrypted skyline tuples.
Accordingly, {\main} leverages a distributed trust architecture for the cloud entity, where the power of the cloud is split into two cloud servers (denoted by $CS_1$ and $CS_2$) hosted by independent cloud service providers.
Note that for convenience of presentation, we will sometimes simply refer to cloud servers $CS_1$ and $CS_2$ as {\csa} in this paper.  
The two cloud servers jointly empower the service of oblivious user-defined skyline query processing.
Such distributed trust architecture has seen growing adoption in various application domains (spanning both academia \cite{liu2019secure,ding2021efficient,zheng2022secskyline,meng2018top,chen2020metal,du2020graphshield,cui2020svknn} and industry \cite{Mozilla,WhitePaper})  and {\main} also follows such trend.

\subsection{Threat Assumptions and Security Guarantees}

In {\main}, similar to prior work \cite{zhang2022efficient}, we consider that the data owner and the user are trusted parties who will adhere to our protocol specifications.
Regarding the cloud servers, we assume a semi-honest adversary model, where each cloud server honestly follows our protocol in providing the user-defined skyline query service, yet may be curious to learn private information from the protocol execution.
Meanwhile, the two cloud servers do not collude with each other.
Such semi-honest and non-colluding threat assumptions on the two cloud servers have been widely adopted as well in other secure systems \cite{liu2019secure,ding2021efficient,zheng2022secskyline,du2020graphshield,chen2020metal}.

Under the above threat model, similar to most existing works in privacy-preserving skyline query processing \cite{liu2019secure,ding2021efficient,zheng2022secskyline}, {\main} aims to guarantee that an individual adversarial cloud server cannot learn 1) the contents of the outsourced database $\mathcal{T}$, user-defined skyline query $Q$, and query result $\mathcal{S}$, 2) the (user-defined) dominance relationships among the database tuples and the number of database tuples which each skyline tuple dominates, and 3) the search access patterns, where the search pattern indicates whether a new skyline query has been issued before and the access pattern reveals which tuples in the database are included in the query result $\mathcal{S}$.

\section{The Design of \main}
\label{sec:Design}

\subsection{Design Overview}
    \label{sec:overview}
    As aforementioned, {\main} builds on the lightweight additive secret sharing technique for encrypting the outsourced database and user-defined skyline query, as well as for securely supporting the required skyline computation.
    So in the beginning the data owner encrypts each tuple $\mathbf{t}$ in its database via arithmetic secret sharing and produces the secret-shared database, denoted by $\llb\mathcal{T}\rrb^A$.
    Then, the secret shares $\langle\mathcal{T}\rangle^A_1$ and $\langle\mathcal{T}\rangle^A_2$ are sent to $CS_1$ and $CS_2$, respectively.

    Next we consider how to properly represent and encrypt a user-defined skyline query $Q$.
    Recall that a user-defined skyline query $Q$ by definition consists of a user-selected dimension subset $\mathcal{B}=\{d_j\}_{j=1}^{k}$, a user-defined constrained region $\mathcal{R}=\{(r_{j,l},r_{j,u})\}_{j=1}^{k}$, and a preference set $\mathcal{P}$ over the selected dimensions in $\mathcal{B}$.
    In order to hide the user's target dimensions and their associated range constraints and preferences while still allowing the query to be processed correctly, {\main} conducts the following to produce an encrypted user-defined skyline query.
    Firstly, the user-defined constrained region $\mathcal{R}$ is extended to $\mathcal{R'}$ that covers all dimensions in $\mathcal{D}$, i.e., $\mathcal{R'}=\{(r'_{i,l},r'_{i,u})\}_{i=1}^m$. If $i=d_j \in \mathcal{B}$, $r'_{i,l}=r_{j,l}$ and $r'_{i,u}=r_{j,u}$; otherwise $r'_{i,l}$ and $r'_{i,u}$ are set to the public lower bound and upper bound respectively on dimension $i$.
    It is easy to see that the tuples falling within the region $\mathcal{R'}$ are identical to those satisfying $\mathcal{R}$.

    Secondly, the preference set $\mathcal{P}$ is extended to $\mathcal{P}'$ that covers all dimensions in $\mathcal{D}$, i.e., $\mathcal{P}'=\{\mathbf{p}'_{i}\}_{i=1}^{m}$, where $\mathbf{p}'_i$ \emph{as per our design} is a 2-bit vector: $\mathbf{p}'_i\in\{00,01,10\}$. 
    Specifically, if $i=d_j \in \mathcal{B}$ and the user prefers minimal value on the $i$-th dimension, then $\mathbf{p}'_i=00$; if $i=d_j \in \mathcal{B}$ and the user prefers maximal value on the $i$-th dimension, then $\mathbf{p}'_i=01$; if $i\neq d_j \in \mathcal{B}$, i.e., the $i$-th dimension is not concerned by the user, then $\mathbf{p}'_i=10$. 
    %
    %
    After the expansion, the user encrypts each $(r'_{i,l},r'_{i,u})\in\mathcal{R}'$  via arithmetic secret sharing, producing the secret sharing of the user-defined constrained region $ \llb\mathcal{R}'\rrb^A$.
   	Additionally, the user encrypts each $\mathbf{p}'_{i}\in\mathcal{P}'$ via binary secret sharing, producing the secret sharing of the preference set $ \llb\mathcal{P}'\rrb^B$.
    The user then sends the secret shares $(\langle \mathcal{R}'\rangle^A_1, \langle \mathcal{P}'\rangle^B_1)$ and
    $(\langle \mathcal{R}'\rangle^A_2, \langle \mathcal{P}'\rangle^B_2)$ to $CS_1$ and $CS_2$, respectively.
    So the encrypted user-defined skyline query is represented as $\llb Q'\rrb=(\llb\mathcal{R}'\rrb^A,\llb\mathcal{P}'\rrb^B)$.

   We then consider how {\csa} can obliviously process the encrypted user-defined skyline query $\llb Q'\rrb$ on the encrypted database $\llb\mathcal{T}\rrb^A$ to produce the encrypted result set $\llb\mathcal{S}\rrb$.
   From a high-level perspective, we build upon the BNL algorithm \cite{borzsonyi2001skyline} (as introduced in Section \ref{sec:BNL}), which follows the state-of-the-art prior work \cite{zhang2022efficient}. 
   However, we observe that directly applying the BNL algorithm in the ciphertext domain without delicate design, as was done in \cite{zhang2022efficient}, will expose private information associated with data patterns including the dominance relationships between database tuples, the number of database tuples dominated by each skyline tuple, and the access pattern (which in turn would disclose the search pattern). 
   
    \begin{algorithm}[t!]
    \caption{High-Level Protocol Overview of {\main}}
    \label{alg:overview}
    \begin{algorithmic}[1]
        \REQUIRE The encrypted user-defined skyline query $\llb Q'\rrb=(\llb\mathcal{R}'\rrb^A,\llb\mathcal{P}'\rrb^B)$ and encrypted original database $[\![{\mathcal{T}}]\!]^A$.
        \ENSURE The encrypted result set of skyline tuples $\llb\mathcal{S}\rrb$.
        
        \STATE $\llb\widetilde{\mathcal{T}}\rrb^A=\mathsf{ObliShuff}(\llb\mathcal{T}\rrb^A)$.
        
        \STATE $\llb\mathcal{C}\rrb^A=\mathsf{ObliGen}(\llb\widetilde{\mathcal{T}}\rrb^A,\llb\mathcal{R}'\rrb^A)$.
        \STATE $\llb\mathcal{S}\rrb=\mathsf{ObliFetch}(\llb\mathcal{C}\rrb^A,\llb\mathcal{P}'\rrb^B)$.
        \RETURN The encrypted result set of skyline tuples $\llb\mathcal{S}\rrb$.
    \end{algorithmic}
\end{algorithm}

   The key challenge here is how to allow the cloud servers to obliviously evaluate the user-defined dominance relationships among database tuples as per the encrypted user-defined skyline query, and proceed correctly without disclosing the dominance relationships. 
   To tackle this challenge, our main insight is to have {\csa} obliviously shuffle $\llb\mathcal{T}\rrb^A$ based on a random permutation unknown to them, prior to the processing of the encrypted user-defined skyline query.
   Subsequently, the encrypted user-defined skyline query is processed over the encrypted \emph{shuffled} database.
   In this case, the dominance relationships among tuples in the (encrypted) shuffled database $\widetilde{\mathcal{T}}$ can be safely revealed, as the permutation is unknown to the cloud servers and they cannot learn the true dominance relationships among tuples in the original database $\mathcal{T}$.
   As will be clear later, such oblivious database shuffling will also greatly facilitate the oblivious finding of database tuples that satisfy the user-defined constrained region.
   More specifically, it will allow the generation of a sub-database that only contains encrypted tuples falling within the user-defined constrained region, while ensuring that the cloud servers cannot learn which tuples in the original database are contained in this sub-database.

    Under such oblivious database shuffling-based framework, we construct the protocol for oblivious user-defined skyline query processing, which consists of three subroutines: (i) oblivious database shuffling, denoted as $\mathsf{ObliShuff}$, (ii) oblivious sub-database generation, denoted as $\mathsf{ObliGen}$, and (iii) oblivious user-defined skyline fetching, denoted as $\mathsf{ObliFetch}$.
    A high-level protocol overview is given in Algorithm \ref{alg:overview}.
    In particular, given the encrypted query $\llb Q'\rrb$ and encrypted database $\llb\mathcal{T}\rrb^A$, {\main} introduces $\mathsf{ObliShuff}$ to have {\csa} obliviously shuffle $\llb\mathcal{T}\rrb^A$ to produce the encrypted shuffled database $\llb\widetilde{\mathcal{T}}\rrb^A$.
    Then, {\main} introduces $\mathsf{ObliGen}$ to have {\csa} obliviously select the encrypted tuples from $\llb\widetilde{\mathcal{T}}\rrb^A$ that fall within the encrypted user-defined constrained region $\llb\mathcal{R}'\rrb^A$, and produce the encrypted sub-database $\llb\mathcal{C}\rrb^A$ of $\llb\widetilde{\mathcal{T}}\rrb^A$.
    Afterwards, {\main} introduces $\mathsf{ObliFetch}$ to have {\csa} obliviously perform the BNL algorithm on $\llb\mathcal{C}\rrb^A$ based the encrypted preference set $ \llb\mathcal{P}'\rrb^B$  to fetch the result set $\llb\mathcal{S}\rrb$.
    In what follows, we elaborate on the design of each subroutine.

\subsection{Oblivious Database Shuffling}
    \label{sec:shuff}

	We first introduce how the encrypted database $\llb\mathcal{T}\rrb^A$ is obliviously shuffled by {\csa} in {\main}. 
	The purpose of this subroutine is to obliviously shuffle the tuples in $\llb\mathcal{T}\rrb^A$ (treated as a matrix $\llb\mathbf{T}\rrb^A$ where each row corresponds to a tuple) through a secret permutation unknown to the cloud servers.
    Since {\main} works with secret-shared data, what we need is an oblivious shuffle protocol that can work in the secret sharing domain.
    We observe that the state-of-the-art oblivious shuffle protocol from \cite{eskandarian2022clarion} is a good fit, because it works with secret-shared inputs and outputs under a two-party setting.
	We adapt this protocol to realize the subroutine of oblivious database shuffling $\mathsf{ObliShuff}$ in {\main} as follows.
	$\mathsf{ObliShuff}$ starts with $CS_1$ holding $\langle\mathbf{T}\rangle^A_1$, $\mathbf{A}_1$, $\mathbf{B}$ and $\pi_1$, and $CS_2$ holding $\langle\mathbf{T}\rangle^A_2$, $\mathbf{A}_2$, $\pi_2$ and $\Delta$.
	Here, $\mathbf{A}_1$, $\mathbf{A}_2$ and $\mathbf{B}$ are random matrices of the same size as $\mathbf{T}$; $\pi_1$ and $\pi_2$ are random permutations over $\{1,2,\cdots,n\}$ (where $n$ is the number of tuples in $\mathbf{T}$); $\Delta$ is a matrix such that $\Delta=\pi_2(\pi_1(\mathbf{A}_2)+\mathbf{A}_1)-\mathbf{B}$, where the permutations are applied in a rowwise manner.
	Note that the input-independent shuffle correlation quantities (i.e., $\mathbf{A}_1$, $\mathbf{A}_2$, $\mathbf{B}$, $\pi_1$, $\pi_2$, and $\Delta$) can be prepared in advance and distributed offline by a third party, such as the data owner.

	$\mathsf{ObliShuff}$ then works as follows to produce the encrypted shuffled database $\llb\widetilde{\mathbf{T}}\rrb^A$ for $\llb\mathbf{T}\rrb^A$.
	%
	%
	Firstly, $CS_2$ masks its share $\langle\mathbf{T}\rangle^A_2$ with $\mathbf{A}_2$:
	$\mathbf{Z}_2=\langle\mathbf{T}\rangle^A_2 - \mathbf{A}_2$, and sends $\mathbf{Z}_2$ to $CS_1$.
	Upon receiving $\mathbf{Z}_2$, $CS_1$ constructs $\mathbf{Z}_1=\pi_1(\mathbf{Z}_2+\langle\mathbf{T}\rangle^A_1)-\mathbf{A}_1$, and sends $\mathbf{Z}_1$ to $CS_2$.
	Afterwards, $CS_2$ locally computes $\pi_2(\mathbf{Z}_1)+\Delta$.
	Finally, $CS_1$ and $CS_2$ respectively hold the shares of $\llb\widetilde{\mathbf{T}}\rrb^A$: $\langle\widetilde{\mathbf{T}}\rangle^A_1=\mathbf{B}$ and $\langle\widetilde{\mathbf{T}}\rangle^A_2=\pi_2(\mathbf{Z}_1)+\Delta$.
	The above procedure can correctly produce a shuffled version of $\llb\mathbf{T}\rrb^A$, since
	\begin{align*}
		\langle\widetilde{\mathbf{T}}\rangle^A_1+\langle\widetilde{\mathbf{T}}\rangle^A_2 &= \mathbf{B}+\pi_2(\mathbf{Z}_1)+\Delta \\
		&= \mathbf{B}+\pi_2(\pi_1(\mathbf{Z}_2+\langle\mathbf{T}\rangle^A_1)-\mathbf{A}_1)+\Delta \\
		&= \mathbf{B}+\pi_2(\pi_1(\mathbf{T}-\mathbf{A}_2)-\mathbf{A}_1)+\Delta \\
		&= \pi_2(\pi_1(\mathbf{T}-\mathbf{A}_2)-\mathbf{A}_1)+\pi_2(\pi_1(\mathbf{A}_2)+\mathbf{A}_1) \\
		&= \pi_2(\pi_1(\mathbf{T})).
	\end{align*}
	Note that each $CS_i$ only knows its local permutation $\pi_i$ ($i\in \{1,2\}$), and is oblivious to the joint permutation $\pi=\pi_2(\pi_1(\cdot))$ that is actually applied to the encrypted database $\llb\mathbf{T}\rrb^A$.
    %


	\label{sec:scan}
    \begin{algorithm}[t!]
        \caption{Oblivious Sub-Database Generation $\mathsf{ObliGen}$}
        \label{alg:scan}
        \begin{algorithmic}[1]
            \REQUIRE The encrypted shuffled database $\llb\widetilde{\mathcal{T}}\rrb^A$ and the encrypted user-defined constrained region $ \llb\mathcal{R}'\rrb^A$.
            
            \ENSURE The encrypted sub-database $\llbracket\mathcal{C}\rrbracket^A$.

            \STATE Initialize $\llb\mathcal{C}\rrb^A=\emptyset$.
            \FORALL{$\llb\widetilde{\mathbf{t}}_i\rrb^A\in\llb\widetilde{\mathcal{T}}\rrb^A$}
                \FORALL{$(\llb r'_{j,l}\rrb^A,\llb r'_{j,u}\rrb^A)\in\llb\mathcal{R}'\rrb^A,j\in[1,m]$}
                    \STATE $\llb\alpha\rrb^B=\mathsf{secLEQ}(\llb r'_{j,l}\rrb^A,\llb\widetilde{\mathbf{t}}_i[j]\rrb^A)$.
                    \STATE $\llb\beta\rrb^B=\mathsf{secLEQ}(\llb\widetilde{\mathbf{t}}_i[j]\rrb^A,\llb r'_{j,u}\rrb^A)$.
                    \STATE $\llb\delta_j\rrb^B=\llb\alpha\rrb^B\otimes\llb\beta\rrb^B$.
                \ENDFOR 
                \STATE $\llb\widehat{\delta}_i\rrb^B=\llb\delta_1\rrb^B\otimes\cdots\otimes\llb\delta_m\rrb^B$.
                \STATE Safely open $\widehat{\delta}_i$.\label{line:opend}
                \IF{$\widehat{\delta}_i=1$}
                    \STATE $\llb\mathcal{C}\rrb^A.append(\llb\widetilde{\mathbf{t}}_i\rrb^A)$.
                \ENDIF
            \ENDFOR
            \RETURN The encrypted sub-database $\llb\mathcal{C}\rrb^A$.
        \end{algorithmic}
    \end{algorithm}

\subsection{Oblivious Sub-Database Generation}

    After obtaining the encrypted shuffled database $\llb\widetilde{\mathcal{T}}\rrb^A$ (i.e., the set representation of $\llb\widetilde{\mathbf{T}}\rrb^A$), {\csa} then need to select the encrypted tuples from $\llb\widetilde{\mathcal{T}}\rrb^A$ that fall within the encrypted user-defined constrained region $\llb\mathcal{R}'\rrb^A$, and produce the encrypted sub-database $\llb\mathcal{C}\rrb^A$.
    We present the complete construction of oblivious sub-database generation $\mathsf{ObliGen}$ in Algorithm \ref{alg:scan}.
    Note that given an encrypted tuple $\llb\widetilde{\mathbf{t}}_i\rrb^A\in\llb\widetilde{\mathcal{T}}\rrb^A$, obliviously evaluating whether it falls within $\llb\mathcal{R}'\rrb^A$ requires {\csa} to securely evaluate $\llb r'_{j,l}\rrb^A\leq\llb\widetilde{\mathbf{t}}_i[j]\rrb^A\leq\llb r'_{j,u}\rrb^A,j\in[1,m]$, where $(\llb r'_{j,l}\rrb^A,\llb r'_{j,u}\rrb^A)\in\llb\mathcal{R}'\rrb^A$.
    This needs secure comparison to be conducted over secret-shared values, and we note that it can be achieved by using a secure most significant bit (MSB) extraction protocol $\mathsf{SecExt}(\cdot)$ \cite{liu2021medisc}.
    Specifically, given two secret-shared values $\llb a\rrb^A$ and $\llb b\rrb^A$, $\mathsf{SecExt}(\llb a\rrb^A,\llb b\rrb^A)$ can securely extract the encrypted MSB of $(a-b)$. 
    If $a<b$, $\mathsf{SecExt}(\llb a\rrb^A,\llb b\rrb^A)$ evaluates to $\llb 1\rrb^B$, and otherwise $\llb 0\rrb^B$.
    For more details on the construction $\mathsf{SecExt}(\cdot)$, we refer the readers to \cite{liu2021medisc}.
    So we can build on the protocol $\mathsf{SecExt}(\cdot)$ for secure comparison.
    However, as $\mathsf{SecExt}(\cdot)$ securely evaluates the operator ``$<$'' instead of the operator ``$\leq$'', we tailor $\mathsf{secLEQ}(\llb a\rrb^A,\llb b\rrb^A)=\neg\mathsf{SecExt}(\llb b\rrb^A,\llb a\rrb^A)$ to suit our context.
    It is easy to see that if $a\leq b$, $\mathsf{secLEQ}(\llb a\rrb^A,\llb b\rrb^A)=\llb 1\rrb^B$, and otherwise $\mathsf{secLEQ}(\llb a\rrb^A,\llb b\rrb^A)=\llb 0\rrb^B$.

    We then introduce the design intuition of $\mathsf{ObliGen}$, which relies on $\mathsf{secLEQ}(\cdot)$.
    Given an encrypted tuple $\llb\widetilde{\mathbf{t}}_i\rrb^A\in\llb\widetilde{\mathcal{T}}\rrb^A$, {\main} has {\csa} evaluate the following:
    \begin{align}\notag
    	    \llb\alpha\rrb^B&=\mathsf{secLEQ}(\llb r'_{j,l}\rrb^A,\llb\widetilde{\mathbf{t}}_i[j]\rrb^A),\\\notag
    	 \llb\beta\rrb^B&=\mathsf{secLEQ}(\llb\widetilde{\mathbf{t}}_i[j]\rrb^A,\llb r'_{j,u}\rrb^A),\\\notag
    	 \llb\delta_j\rrb^B&=\llb\alpha\rrb^B\otimes\llb\beta\rrb^B,j\in[1,m],
    \end{align}
    where $\delta_j=1$ indicates that $\widetilde{\mathbf{t}}_i[j]$ satisfies the range constraint $ [r'_{j,l}, r'_{j,u}]$, and $\delta_j=0$ indicates the opposite.
    After that, aggregation of $\llb\delta_j\rrb^B$, $j\in[1,m]$ is performed to produce $\llb\widehat{\delta}_i\rrb^B$ by:
    \begin{equation}\notag
    	\llb\widehat{\delta}_i\rrb^B=\llb\delta_1\rrb^B\otimes\llb\delta_2\rrb^B\cdots\otimes\llb\delta_m\rrb^B,
    \end{equation}
    where $\widehat{\delta}_i=1$ indicates that tuple $\widetilde{\mathbf{t}}_i$ falls within the constrained region $\mathcal{R}'$, and $\widehat{\delta}_i=0$ indicates not.
    After that, $\widehat{\delta}_i$ can be \emph{safely} opened by {\csa} to determine whether adding $\llb\widetilde{\mathbf{t}}_i\rrb^A$ into the encrypted sub-database $\llb\mathcal{C}\rrb^A$.
	Note that since $\mathsf{ObliShuff}$ obliviously breaks the mappings between tuples in the encrypted original database $\llb\mathcal{T}\rrb^A$ and tuples in the encrypted shuffled database $\llb\widetilde{\mathcal{T}}\rrb^A$, {\csa} cannot learn which tuples in $\llb\mathcal{T}\rrb^A$ are added into $\llb\mathcal{C}\rrb^A$.

\subsection{Oblivious User-Defined Skyline Fetching}
    \label{sec:fetch}

	So far we have introduced how {\csa} obliviously select the encrypted tuples falling within the encrypted constrained region $\llb\mathcal{R}'\rrb^A$ and generate the encrypted sub-database $\llb\mathcal{C}\rrb^A$. 
	Now we introduce $\mathsf{ObliFetch}$ to enable {\csa} to obliviously execute the BNL algorithm and fetch the encrypted result set of skyline tuples $\llb\mathcal{S}\rrb$ from $\llb\mathcal{C}\rrb^A$.
    In the context of user-defined skyline queries, we observe that the primary challenge in securely realizing the BNL algorithm is evaluating the encrypted user-defined dominance $\llb\Phi\rrb^B=\llb\mathbf{a}\preceq_{\mathcal{B}}\mathbf{b}\rrb^B$ (as per Definition \ref{def:1}), given two encrypted tuples $\llb\mathbf{a}\rrb^A$ and $\llb\mathbf{b}\rrb^A$, as well as the encrypted preference set $\llb\mathcal{P}'\rrb^B$. 
    If $\mathbf{a}\preceq_{\mathcal{B}}\mathbf{b}$ holds, $\Phi=1$; otherwise, $\Phi=0$.

         \begin{algorithm}[t!]
        \caption{Oblivious User-Defined Dominance Determination $\mathsf{ObliDom}$}
        \label{alg:dom}
        \begin{algorithmic}[1]
            \REQUIRE The encrypted tuples $\llb\mathbf{a}\rrb^A$ and $\llb\mathbf{b}\rrb^A$, and the encrypted preference set  $\llb\mathcal{P}'\rrb^B=\{\llb\textbf{p}_i'\rrb^B\}_{i=1}^{m}$.

            \ENSURE The encrypted user-defined dominance flag $\llb\Phi\rrb^B$.

            \FORALL{$\llb\textbf{p}_i'\rrb^B\in\llb\mathcal{P}'\rrb^B,i\in[1,m]$}
                \STATE $\llb\alpha_i\rrb^B=\mathsf{secLEQ}(\llb\mathbf{a}[i]\rrb^A,\llb\mathbf{b}[i]\rrb^A)$. //Evaluate the magnitude relation between $\mathbf{a}[i]$ and $\mathbf{b}[i]$.
                
                \STATE $\llb\alpha'_i\rrb^B=\mathsf{secLEQ}(\llb\mathbf{b}[i]\rrb^A,\llb\mathbf{a}[i]\rrb^A)$.
                
                \STATE $\llb\beta_i\rrb^B=\neg\llb\textbf{p}_i'[2]\rrb^B\otimes\llb\alpha_i\rrb^B$.
                
                \STATE $\llb\beta'_i\rrb^B=\llb\textbf{p}_i'[2]\rrb^B\otimes\llb\alpha'_i\rrb^B$.
                
                \STATE $\llb\phi_i\rrb^B=\llb\beta_i\rrb^B\oplus\llb\beta'_i\rrb^B$.  //Evaluate whether the magnitude relation between $\mathbf{a}[i]$ and $\mathbf{b}[i]$ is consistent with the user preference on the $i$-th dimension.
                
                \STATE $\llb\sigma_i\rrb^B=\llb\phi_i\rrb^B\vee\llb\textbf{p}_i'[1]\rrb^B$.
                
                \STATE $\llb\omega_i\rrb^B=\neg(\llb\alpha_i\rrb^B\otimes\llb\alpha_i'\rrb^B)\otimes\neg\llb\textbf{p}_i'[1]\rrb^B$. //Evaluate whether $\mathbf{a}[i]\neq\mathbf{b}[i]$ and whether the user is concerned with the $i$-th dimension.

            \ENDFOR
            
            //Aggregate the encrypted evaluation results for each dimension:
            
            \STATE $\llb\widehat{\sigma}\rrb^B=\llb\sigma_1\rrb^B\otimes\cdots\otimes\llb\sigma_m\rrb^B$.
            
           \STATE$\llb\widehat{\omega}\rrb^B=\llb\omega_1\rrb^B\vee\cdots\vee\llb\omega_m\rrb^B$.

            \STATE $\llb\Phi\rrb^B=\llb\widehat{\sigma}\rrb^B\otimes\llb\widehat{\omega}\rrb^B$.

            \RETURN The encrypted dominance flag  $\llb\Phi\rrb^B$. 
        \end{algorithmic}
    \end{algorithm}

    \noindent\textbf{Oblivious user-defined dominance evaluation.}
    In order to tackle this challenge, we propose an oblivious user-defined dominance evaluation protocol $\mathsf{ObliDom}$, as given in Algorithm \ref{alg:dom}, which computes $\llb\Phi\rrb^B=\mathsf{ObliDom}(\llb\mathbf{a}\rrb^A,\llb\mathbf{b}\rrb^A,\llb\mathcal{P}'\rrb^B)$.
    At a high level, $\mathsf{ObliDom}$ first securely evaluates the dominance relationships between each dimension of $\llb\mathbf{a}\rrb^A$ and $\llb\mathbf{b}\rrb^A$ based on the corresponding encrypted preference information $\llb\textbf{p}_i'\rrb^B,i\in[1,m]$ (i.e., lines 1-9 of Algorithm \ref{alg:dom}), and then securely aggregates the encrypted evaluation results to produce the final encrypted dominance flag $\llb\Phi\rrb^B$ (i.e., lines 10-12 of Algorithm \ref{alg:dom}).
    The underlying idea behind the design of $\mathsf{ObliDom}$ is introduced as follows.

    Given $\llb\mathbf{a}[i]\rrb^A$ and $\llb\mathbf{b}[i]\rrb^A, i\in[1,m]$, {\csa} first evaluate:
    \begin{align}\notag
    	\llb\alpha_i\rrb^B=\mathsf{secLEQ}(\llb\mathbf{a}[i]\rrb^A,\llb\mathbf{b}[i]\rrb^A),\\\notag
    	 \llb\alpha_i'\rrb^B=\mathsf{secLEQ}(\llb\mathbf{b}[i]\rrb^A,\llb\mathbf{a}[i]\rrb^A),
    \end{align} 
	where $\alpha_i=1$ indicates that $\mathbf{a}[i]\leq\mathbf{b}[i]$ and $\alpha'_i=1$ indicates that $\mathbf{b}[i]\leq\mathbf{a}[i]$. 
	Conversely, if $\alpha_i=0$ or $\alpha'_i=0$, it signifies the opposite.
    Afterwards, {\csa} evaluate:
    \begin{align}\notag
    	 \llb\beta_i\rrb^B&=\neg\llb\textbf{p}_i'[2]\rrb^B\otimes\llb\alpha_i\rrb^B,\\\notag
    	 \llb\beta'_i\rrb^B&=\llb\textbf{p}_i'[2]\rrb^B\otimes\llb\alpha'_i\rrb^B,\\\notag
    	 \llb\phi_i\rrb^B&=\llb\beta_i\rrb^B\oplus\llb\beta'_i\rrb^B,
    \end{align}
    where $\textbf{p}_i'[2]$ is the second bit of the preference vector $\textbf{p}_i'$.
    Recall that $\textbf{p}_i'[2]=0$ means that the user prefers minimal value on the $i$-th dimension; $\textbf{p}_i'[2]=1$ means that the user prefers maximal value on the $i$-th dimension.
    Here, $\phi_i=1$ indicates that the magnitude relation between $\mathbf{a}[i]$ and $\mathbf{b}[i]$ is consistent with the user preference on the $i$-th dimension.

    We analyze the correctness as follows.
    $\phi_i=1$ means that $\beta_i=1$ or $\beta'_i=1$.
    If $\beta_i=1, \beta'_i=0$, then $\textbf{p}_i'[2]=0$ and $\alpha_i=1$, namely the user prefers minimal value on the $i$-th dimension and $\mathbf{a}[i]\leq\mathbf{b}[i]$.
    If $\beta_i=0, \beta'_i=1$, then $\textbf{p}_i'[2]=1$ and $\alpha'_i=1$, namely the user prefers maximal value on the $i$-th dimension and $\mathbf{b}[i]\leq\mathbf{a}[i]$.

    However, the above evaluation does not involve the condition regarding whether the user indeed chooses the $i$-th dimension or not.
    Therefore, {\csa} further evaluate:
    \begin{equation}\notag
   \llb\sigma_i\rrb^B=\llb\phi_i\rrb^B\vee\llb\textbf{p}_i'[1]\rrb^B,
    \end{equation}
  	where $\textbf{p}_i'[1]$ is the first bit of the preference vector $\textbf{p}_i'$.
    Recall that $\textbf{p}_i'[1]=1$ means that the user does not actually choose the $i$-th dimension, and $\textbf{p}_i'[1]=0$ means the opposite. 
  	 $\vee$ denotes the ``OR'' operation.
  	 Given two secret-shared bits $\llb b_1\rrb^B$ and $\llb b_2\rrb^B$, the secure OR operation over them can be performed by $\llb b_1\vee b_2\rrb^B=\neg(\neg\llb b_1\rrb^B \otimes \neg\llb  b_b\rrb^B)$.
  	Clearly, if the user indeed does not choose the $i$-th dimension, then $\textbf{p}_i'[1]=1$ and $\sigma_i=1$ always holds.
  	Subsequently, {\csa} aggregate the results on all dimensions by:
  	\begin{equation}\notag
  		\llb\widehat{\sigma}\rrb^B=\llb\sigma_1\rrb^B\otimes\cdots\otimes\llb\sigma_m\rrb^B, 
  	\end{equation}
  	where $\widehat{\sigma}=1$ indicates that $\mathbf{a}\preceq_{\mathcal{B}}\mathbf{b}$ or $\mathbf{a}=\mathbf{b}$, and $\widehat{\sigma}=0$ indicates the opposite.
  	Therefore, {\csa} need to further obliviously evaluate whether $\mathbf{a}=\mathbf{b}$ holds to obtain the final $\llb\Phi\rrb^B$, i.e., the final encrypted dominance evaluation result.

  	Our main insight is to let {\csa} evaluate:
  	\begin{equation}\notag
  		\llb\omega_i\rrb^B=\neg(\llb\alpha_i\rrb^B\otimes\llb\alpha_i'\rrb^B)\otimes\neg\llb\textbf{p}_i'[1]\rrb^B,
  	\end{equation}
  	where $\omega_i=1$ indicates that $\mathbf{a}[i]\neq\mathbf{b}[i]$ and the user is concerned with the $i$-th dimension.
  	We analyze the correctness as follows.
  	Firstly, $\omega_i=1$ means that $\alpha_i\otimes\alpha_i'=0$, i.e., $\alpha_i=1$ and $\alpha_i'=1$ cannot hold at the same time. So $\mathbf{a}[i]\leq\mathbf{b}[i]$ and $\mathbf{b}[i]\leq\mathbf{a}[i]$ cannot hold at the same time, i.e., $\mathbf{a}[i]\neq\mathbf{b}[i]$.
  	Secondly, $\omega_i=1$ also means $\textbf{p}_i'[1]=0$, i.e., the user is concerned with the $i$-th dimension.
  	Afterwards, {\csa} aggregate $\llb\omega_i\rrb^B,i\in[1,m]$ by
  	\begin{equation}\notag
  		\llb\widehat{\omega}\rrb^B=\llb\omega_1\rrb^B\vee\cdots\vee\llb\omega_m\rrb^B,
  	\end{equation} 
  	where $\widehat{\omega}=1$ indicates that $\exists i\in[1,m],\mathbf{a}[i]\neq\mathbf{b}[i]$ and the user is concerned with the $i$-th dimension.
  	Finally, {\csa} aggregate $\llb\widehat{\sigma}\rrb^B$ and $\llb\widehat{\omega}\rrb^B$ to produce
  	\begin{equation}\notag
  	\llb\Phi\rrb^B=\llb\widehat{\sigma}\rrb^B\otimes\llb\widehat{\omega}\rrb^B,
  	\end{equation}
  	where $\Phi=1$ indicates that if the user prefers minimal value on the $i$-th dimension, then $\mathbf{a}[i]\leq\mathbf{b}[i]$; if the user prefers maximal value on the $i$-th dimension, then $\mathbf{b}[i]\leq\mathbf{a}[i]$, and there exists a user-selected dimension $j$ such that $\mathbf{a}[j]\neq\mathbf{b}[j]$.
  	Therefore, $\llb\Phi\rrb^B=\llb\mathbf{a}\preceq_{\mathcal{B}}\mathbf{b}\rrb^B$.

  \begin{algorithm}[t]
  	\caption{Oblivious User-Defined Skyline Fetching $\mathsf{ObliFetch}$}
  	\label{alg:fetch}
  	\begin{algorithmic}[1]
  		\REQUIRE The encrypted sub-database $\llbracket\mathcal{C}\rrbracket^A$ and the encrypted preference set  $\llb\mathcal{P}'\rrb^B$.

  		\ENSURE The encrypted result set $\llb\mathcal{S}\rrb$.


  		\STATE Initialize $\llb\mathcal{S}\rrb=\emptyset$; $\llb\mathsf{isDomi}\rrb^B=\llb0\rrb^B$.

  		\STATE $\llb\mathcal{S}\rrb.append((\llb\widetilde{\mathbf{t}}_1\rrb^A,\llb\mathsf{isDomi}\rrb^B))$. // $\llb\widetilde{\mathbf{t}}_1\rrb^A\in\llbracket\mathcal{C}\rrbracket^A$.

  		\FORALL{$\llb\widetilde{\mathbf{t}}_i\rrb^A\in\llbracket\mathcal{C}\rrbracket^A,i\geq2$}
  		
  		\STATE Initialize $\mathsf{flag}=1$; $\llb\widehat{\Phi}\rrb^B=\llb0\rrb^B$.
  		
  		\FORALL{$(\llb\widetilde{\mathbf{t}}\rrb^A,\llb\mathsf{isDomi}\rrb^B)\in\llb\mathcal{S}\rrb$}
  		
  		\STATE $\llb\Phi_1\rrb^B=\mathsf{ObliDom}(\llb\widetilde{\mathbf{t}}\rrb^A,\llb\widetilde{\mathbf{t}}_i\rrb^A,\llb\mathcal{P}'\rrb^B)$.
  		
  		\STATE $\llb\widehat{\Phi}\rrb^B=\llb\widehat{\Phi}\rrb^B\vee\llb\Phi_1\rrb^B$.

  		\STATE Locally generate a secret-shared random bit $\llb r\rrb^B$. 
  		
  		\STATE $\llb\Phi_1'\rrb^B= \llb\Phi_1\rrb^B\otimes\llb r\rrb^B$. 
  		
  		\STATE Safely open $\Phi'_1$.\label{line:fetchOpen1} 
  		
  		\IF{$\Phi'_1=1$} \label{line:discard1} 
  		\STATE $\mathsf{flag}=0$.
  		
  		\STATE Break. //$\widetilde{\mathbf{t}}$ dominates $\widetilde{\mathbf{t}}_i$. 
  		
  		\ENDIF \label{line:discard2} 
  		
  		\STATE $\llb\Phi_2\rrb^B=\mathsf{ObliDom}(\llb\widetilde{\mathbf{t}}_i\rrb^A,\llb\widetilde{\mathbf{t}}\rrb^A,\llb\mathcal{P}'\rrb^B)$.
  		\STATE Safely open $\Phi_2$. \label{line:fetchOpen2}
  		
  		\IF{$\Phi_2=1$}
  		\STATE $\llb\mathcal{S}\rrb.remove((\llb\widetilde{\mathbf{t}}\rrb^A,\llb\mathsf{isDomi}\rrb^B))$.  //$\widetilde{\mathbf{t}}_i$ dominates $\widetilde{\mathbf{t}}$. 
  		\ENDIF
  		\ENDFOR
  		\IF{$\mathsf{flag}=1$}

  		\STATE     $\llb\mathsf{isDomi}\rrb^B=\llb\widehat{\Phi}\rrb^B$.
  		\STATE $\llb\mathcal{S}\rrb.insert((\llb\widetilde{\mathbf{t}}_i\rrb^A,\llb\mathsf{isDomi}\rrb^B))$. // $\widetilde{\mathbf{t}}_i$ is incomparable with all tuples in $\llb\mathcal{S}\rrb$ or dominated by tuples in $\llb\mathcal{S}\rrb$.
  		\ENDIF
  		\ENDFOR
  		\RETURN The encrypted result set $\llb\mathcal{S}\rrb$.
  	\end{algorithmic}
  \end{algorithm}

	\noindent\textbf{The complete construction of $\mathsf{ObliFetch}$.}
	Algorithm \ref{alg:fetch} presents the complete construction of $\mathsf{ObliFetch}$, which is an oblivious realization of the BNL algorithm.
    %
    %
    At a high level, the encrypted result set $\llb\mathcal{S}\rrb$ is initialized as an empty set and is used to store the encrypted candidate skyline tuples during the execution of $\mathsf{ObliFetch}$.
    For each encrypted tuple $\llb\widetilde{\mathbf{t}}_i\rrb^A\in\llb\mathcal{C}\rrb^A$, $\mathsf{ObliDom}$ is invoked to securely determine the  dominance relationships between $\llb\widetilde{\mathbf{t}}_i\rrb^A$ and each encrypted candidate tuple $\llb\widetilde{\mathbf{t}}\rrb^A$ in $\llb\mathcal{S}\rrb$.
    Subsequently, {\csa} update $\llb\mathcal{S}\rrb$ by securely performing the corresponding operation introduced in Section \ref{sec:BNL} based on the dominance relationships.
    Specifically, if there exists $\llb\widetilde{\mathbf{t}}\rrb^A$ in $\llb\mathcal{S}\rrb$ such that $\llb\widetilde{\mathbf{t}}\rrb^A\preceq_{\mathcal{B}}\llb\widetilde{\mathbf{t}}_i\rrb^A$, then $\llb\widetilde{\mathbf{t}}_i\rrb^A$ is discarded. 
    If there exists $\llb\widetilde{\mathbf{t}}\rrb^A$ in $\llb\mathcal{S}\rrb$ such that $\llb\widetilde{\mathbf{t}}_i\rrb^A\preceq_{\mathcal{B}}\llb\widetilde{\mathbf{t}}\rrb^A$, then $\llb\widetilde{\mathbf{t}}\rrb^A$ is removed from $\llb\mathcal{S}\rrb$. 
    If $\llb\widetilde{\mathbf{t}}_i\rrb^A$ is incomparable with all encrypted tuples in $\llb\mathcal{S}\rrb$, then $\llb\widetilde{\mathbf{t}}_i\rrb^A$ is inserted into $\llb\mathcal{S}\rrb$.
    After iterating over all the tuples in $\llb\mathcal{C}\rrb^A$, {\csa} can obtain the final encrypted result set $\llb\mathcal{S}\rrb$, which is returned to the user for decryption.
    The underlying idea behind the design of $\mathsf{ObliFetch}$ is introduced as follows.

    Given each $\llb\widetilde{\mathbf{t}}_i\rrb^A\in\llb\mathcal{C}\rrb^A$, {\csa} first initialize $\mathsf{flag}=1$ (i.e., line 4),  which is used to indicate whether $\widetilde{\mathbf{t}}_i$ is incomparable with all tuples in the current $\llb\mathcal{S}\rrb$ .
  	Then given each $\llb\widetilde{\mathbf{t}}\rrb^A\in\llb\mathcal{S}\rrb$, {\main} lets {\csa} securely evaluate (i.e., line 6): 
    \begin{equation}\notag
    \llb\Phi_1\rrb^B=\mathsf{ObliDom}(\llb\widetilde{\mathbf{t}}\rrb^A,\llb\widetilde{\mathbf{t}}_i\rrb^A,\llb\mathcal{P}'\rrb^B),
    \end{equation}
	where $\Phi_1=1$ indicates that $\widetilde{\mathbf{t}}$ dominates $\widetilde{\mathbf{t}}_i$, and $\Phi_1=0$ indicates that $\widetilde{\mathbf{t}}$ cannot dominate $\widetilde{\mathbf{t}}_i$.

    After producing the encrypted dominance evaluation result $\llb\Phi_1\rrb^B$, one might think of letting {\csa} open $\Phi_1$ to determine whether to discard $\llb\widetilde{\mathbf{t}}_i\rrb^A$.
    While revealing such dominance evaluation result would not reveal the true dominance relationships among database tuples due to oblivious database shuffling, it has the potential to leak the number of database tuples dominated by each skyline tuple.
    Instead, {\main} lets {\csa} open a \textit{masked version} of $\Phi_1$ (i.e., lines 8-10) to \textit{randomly} discard $\llb\widetilde{\mathbf{t}}_i\rrb^A$.
    Specifically, {\csa} first generate a secret-shared random bit $\llb r\rrb^B$ locally, where $CS_1$ generates a random share $\langle r\rangle^B_1$ and $CS_2$ generates a random share $\langle r\rangle^B_2$.
    After that, {\csa} safely open $\Phi_1'$ instead of $\Phi_1$, where $\llb\Phi_1'\rrb^B= \llb\Phi_1\rrb^B\otimes\llb r\rrb^B$ and $\Phi_1'=1$ indicates that $\widetilde{\mathbf{t}}$ dominates $\widetilde{\mathbf{t}}_i$, and thus {\csa} discard $\llb\widetilde{\mathbf{t}}_i\rrb^A$ by setting $\mathsf{flag}=0$, and then break the current loop to evaluate the next encrypted tuple in $\llbracket\mathcal{C}\rrbracket^A$ (i.e., lines 11-14); $\Phi_1'=0$ indicates that $\widetilde{\mathbf{t}}$ \textit{possibly} does not dominate $\widetilde{\mathbf{t}}_i$, and thus {\csa} do not discard $\llb\widetilde{\mathbf{t}}_i\rrb^A$.
    Therefore, {\csa} do not necessarily discard a dominated tuple, and thus cannot learn the exact number of database tuples which each skyline tuple dominates. 
    Meanwhile, if $\Phi_1=0$, then $\Phi_1'=0$ must hold, which prevents {\csa} from falsely discarding a skyline tuple.

    Afterwards, if $\Phi_1'=0$, {\csa} continue to evaluate whether $\widetilde{\mathbf{t}}_i$ dominates $\widetilde{\mathbf{t}}$ by performing (i.e., line 15) 
       \begin{equation}\notag
    	\llb\Phi_2\rrb^B=\mathsf{ObliDom}(\llb\widetilde{\mathbf{t}}_i\rrb^A,\llb\widetilde{\mathbf{t}}\rrb^A,\llb\mathcal{P}'\rrb^B).
    \end{equation}
    {\csa} then safely open $\Phi_2$, where $\Phi_2=1$ indicates that $\widetilde{\mathbf{t}}_i$ dominates $\widetilde{\mathbf{t}}$, and thus {\csa} remove  $\widetilde{\mathbf{t}}$ from $\llb\mathcal{S}\rrb$.
    After iterating over all the tuples in the current $\llb\mathcal{S}\rrb$ (i.e., lines 5-20), if $\mathsf{flag}=1$ still holds, it indicates that $\widetilde{\mathbf{t}}_i$ is incomparable with all tuples in $\llb\mathcal{S}\rrb$, and therefore {\csa} insert the tuple $\llb\widetilde{\mathbf{t}}_i\rrb^A$ into $\llb\mathcal{S}\rrb$ (i.e., lines 21-24).

    The above process, however, will result in that $\llb\mathcal{S}\rrb$ contains encrypted tuples that are dominated by other tuple(s), as what {\csa} open is a masked version of $\Phi_1$.
    Therefore, {\main} introduces an encrypted bit $\llb\mathsf{isDomi}\rrb^B$ for each tuple inserted into $\llb\mathcal{S}\rrb$. The purpose of introducing this bit $\mathsf{isDomi}$ is to allow the user to later filter out fake skyline tuples from $\mathcal{S}$.
    Every time {\csa} insert a tuple $\llb\widetilde{\mathbf{t}}_i\rrb^A$ into $\llb\mathcal{S}\rrb$, {\main} lets them additionally insert $\llb\mathsf{isDomi}\rrb^B$ corresponding to the tuple into $\llb\mathcal{S}\rrb$, i.e., $\llb\mathcal{S}\rrb.insert(\llb\widetilde{\mathbf{t}}_i\rrb^A,\llb\mathsf{isDomi}\rrb^B)$ at line 23, where $\mathsf{isDomi}=1$ indicates that $\widetilde{\mathbf{t}}_i$ is dominated by at least one tuple in the current $\llb\mathcal{S}\rrb$ before it is inserted; whereas $\mathsf{isDomi}=0$  indicates that $\widetilde{\mathbf{t}}_i$ cannot be dominated by any tuple in the current $\llb\mathcal{S}\rrb$ before it is inserted.
    By this way, when the user receives the final result set $\llb\mathcal{S}\rrb$, it can decrypt it to easily filter out the fake skyline tuples from $\mathcal{S}$ by checking their flags $\mathsf{isDomi}$. 
    Specifically, $\mathsf{isDomi}=1$ indicates that the corresponding skyline tuple is fake, and the user filters it out from $\mathcal{S}$; whereas $\mathsf{isDomi}=0$ indicates that the corresponding skyline tuple is true, and thus the user retains it.

    Next, we introduce how {\csa} securely compute the encrypted bit $\llb\mathsf{isDomi}\rrb^B$ for each tuple inserted into $\llb\mathcal{S}\rrb$.
    Given an encrypted tuple $\llb\widetilde{\mathbf{t}}_i\rrb^A\in\llbracket\mathcal{C}\rrbracket^A$, {\csa} first initialize $\llb\widehat{\Phi}\rrb^B=\llb0\rrb^B$ (i.e., line 1).
    Then after each evaluation of $\llb\Phi_1\rrb^B=\mathsf{ObliDom}(\llb\widetilde{\mathbf{t}}\rrb^A,\llb\widetilde{\mathbf{t}}_i\rrb^A,\llb\mathcal{P}'\rrb^B)$ at line 6, {\csa} securely aggregate $\llb\Phi_1\rrb^B$ and $\llb\widehat{\Phi}\rrb^B$ by (i.e., line 7) 
    \begin{equation}\notag
    	\llb\widehat{\Phi}\rrb^B=\llb\widehat{\Phi}\rrb^B\vee\llb\Phi_1\rrb^B.
    \end{equation}
	When iterating over all the tuples in the current $\llb\mathcal{S}\rrb$ and inserting $\llb\widetilde{\mathbf{t}}_i\rrb^A$ into $\llb\mathcal{S}\rrb$, {\csa} set $\llb\mathsf{isDomi}\rrb^B=\llb\widehat{\Phi}\rrb^B$ and insert $(\llb\widetilde{\mathbf{t}}_i\rrb^A,\llb\mathsf{isDomi}\rrb^B)$ into $\llb\mathcal{S}\rrb$ (i.e., lines 22 and 23), where $\mathsf{isDomi}=1$ means that there exists $\llb\widetilde{\mathbf{t}}\rrb^A$ in $\llb\mathcal{S}\rrb$, such that the evaluation of $\mathsf{ObliDom}(\llb\widetilde{\mathbf{t}}\rrb^A,\llb\widetilde{\mathbf{t}}_i\rrb^A,\llb\mathcal{P}'\rrb^B)$ outputs $\llb\Phi_1\rrb^B=\llb1\rrb^B$, namely, $\llb\widetilde{\mathbf{t}}_i\rrb^A$ is dominated by at least one tuple in the sub-database $\llbracket\mathcal{C}\rrbracket^A$.
    Therefore, each fake skyline tuple in the result set $\mathcal{S}$ can be correctly flagged by its $\mathsf{isDomi}$, i.e., $\mathsf{isDomi}$ corresponding to the fake skyline tuple is 1.

    \section{Security Analysis}
    \label{sec:Security}

    Our security analysis of {\main} follows the standard simulation-based paradigm \cite{lindell2017how}.
    %
    %
	To start, we define the ideal functionality $\mathcal{F}$ for oblivious user-defined skyline query processing.
    \begin{itemize}
        \item \textbf{Input.} The data owner provides the database
        $\mathcal{T}$ to $\mathcal{F}$; the user submits a user-defined skyline query  $Q$ to $\mathcal{F}$; and {\csa} input nothing to $\mathcal{F}$.

        \item \textbf{Computation.} Upon receiving $\mathcal{T}$ and $Q$,
        $\mathcal{F}$ performs the user-defined skyline query processing, and produces the skyline tuples $\mathcal{S}$.
    
        \item \textbf{Output.} $\mathcal{F}$ returns $\mathcal{S}$ to the user.
    \end{itemize}
    Let $\prod$ denote a protocol for oblivious user-defined skyline query processing realizing the ideal functionality $\mathcal{F}$.
    The security of $\prod$ is formally defined as follows:
    
    \begin{definition}
        \label{def:secure}
        Let $\mathsf{View}^{\prod}_{CS_t},t\in\{1,2\}$ denote $CS_t$'s view
        during $\prod$'s execution.
        $\prod$ is secure under the semi-honest and non-colluding adversary model, if for each corrupted $CS_t$, there exists a probabilistic polynomial-time (PPT) simulator $S$ that can generate a simulated view $\textsf{Sim}_{CS_t}$ such that $\textsf{Sim}_{CS_t}$ is indistinguishable from $\mathsf{View}^{\prod}_{CS_t}$.
    \end{definition}

    \begin{theorem}
        \label{theo:secure}
        In the semi-honest and non-colluding adversary model, {\main} can securely realize the ideal functionality $\mathcal{F}$ according to Definition \ref{def:secure}.
    \end{theorem}

    \begin{proof}
        Recall that there are three secure subroutines in {\main}: 1) oblivious database shuffling $\mathsf{ObliShuff}$; 2) oblivious sub-database generation $\mathsf{ObliGen}$; 3) oblivious user-defined skyline fetching $\mathsf{ObliFetch}$.
        {\main} invokes these secure subroutine in order, each of which takes as input secret-shared inputs and outputs.
        %
        Hence, if the simulator for each component exists, the simulator for the whole protocol exists \cite{canetti2000security,katz2005handling,curran2019procsa}.
        Let $\mathsf{Sim}^{\mathtt{X}}_{CS_{t}}$ denote the simulator that generates $CS_{t}$'s view in the execution of secure subroutine $\mathtt{X}$. The construction of the simulator for each secure subroutine is shown as follows.

        \begin{itemize}
            \item $\mathsf{Sim}^{\mathsf{ObliShuff}}_{CS_t}$.
            $\mathsf{Sim}^{\mathsf{ObliShuff}}_{CS_t}$ can be trivially constructed by invoking the simulator for the oblivious shuffle protocol \cite{eskandarian2022clarion}.
            Therefore, according to the security of oblivious shuffle \cite{eskandarian2022clarion}, $\mathsf{Sim}^{\mathsf{ObliShuff}}_{CS_t}$ exists.

            \item $\mathsf{Sim}^{\mathsf{ObliGen}}_{CS_t}$.
            Note that the operations in the secure component $\mathsf{ObliGen}$ (i.e., Algorithm \ref{alg:scan})  that require interaction between {\csa} include the $\mathsf{secLEQ}$ operation, basic secure operations in binary secret sharing domain, and the operations of opening secret-shared values.
            We then analyze the simulators of them in turn. 
            The $\mathsf{secLEQ}$ operation is built on the secure MSB extraction protocol \cite{liu2021medisc}.
            Since secure MSB extraction involves basic binary secret sharing operations (i.e., $\oplus$ and $\otimes$) \cite{liu2021medisc}, it is clear that its simulator exists.
            In addition, $\mathsf{Sim}^{\mathsf{ObliGen}}_{CS_t}$ can also simulate the sequence of opened values $\widehat{\delta}_i,i\in[1,n]$ for the encrypted tuples in $\llb\widetilde{\mathcal{T}}\rrb^A$.
            Specifically, $\mathsf{Sim}^{\mathsf{ObliGen}}_{CS_t}$ legitimately holds the size information of the output sub-database (i.e., the number of $1s$ in $\widehat{\delta}_i,i\in[1,n]$). 
            Moreover, as $\llb\widetilde{\mathcal{T}}\rrb^A$ is produced from $\llb\mathcal{T}\rrb^A$ with a secret random permutation, the positions of $1s$ are random in the revealed sequence $\widehat{\delta}_1,\widehat{\delta}_2,\cdots, \widehat{\delta}_n$. So $\mathsf{Sim}^{\mathsf{ObliGen}}_{CS_t}$ can first generate a zero sequence of length $n$, and then randomly select $|\mathcal{C}|$ $0s$ in the sequence to replace them with $1s$.
            %

            \item $\mathsf{Sim}^{\mathsf{ObliFetch}}_{CS_t}$.
            Note that the operations in the secure component $\mathsf{ObliFetch}$ (i.e., Algorithm \ref{alg:fetch}) that require interaction between {\csa} include the $\mathsf{ObliDom}$ operation (given in Algorithm \ref{alg:dom}), basic secure operations in binary secret sharing domain, and the operations of opening secret-shared values.
            Since $\mathsf{ObliDom}$ consists of $\mathsf{secLEQ}$ and the basic operations (i.e., $\oplus$,  $\otimes$, $\vee$, and $\neg$) in the binary secret sharing domain, the simulator for $\mathsf{ObliDom}$ exists. 
            In addition, during the execution of $\mathsf{ObliFetch}$, two sequences of $\Phi_1'$ and $\Phi_2$ will be generated (i.e., lines \ref{line:fetchOpen1} and \ref{line:fetchOpen2} in Algorithm \ref{alg:fetch}).
            So we need to prove that $\mathsf{Sim}^{\mathsf{ObliFetch}}_{CS_t}$ can  simulate the two sequences. 
            We first prove that $\mathsf{Sim}^{\mathsf{ObliFetch}}_{CS_t}$ can simulate the sequence of $\Phi_1'$.
            Specifically, $\mathsf{Sim}^{\mathsf{ObliFetch}}_{CS_t}$ legitimately holds the number (denoted as $n_d$) of discard operations (i.e., lines \ref{line:discard1}-\ref{line:discard2} in Algorithm \ref{alg:fetch}) carried out during the execution of  $\mathsf{ObliFetch}$.
            This number $n_d$ corresponds to the number of $1s$ in the revealed sequence of $\Phi_1'$.
            Moreover, as $\llbracket\mathcal{C}\rrbracket^A$ is the sub-database of $\llb\widetilde{\mathcal{T}}\rrb^A$, which is produced from $\llb\mathcal{T}\rrb^A$ with a secret random permutation, the positions of $1s$ are random in the revealed sequence of $\Phi_1'$. 
            So $\mathsf{Sim}^{\mathsf{ObliFetch}}_{CS_t}$ can first generate a zero sequence that is the same length as the generated sequence of $\Phi_1'$ during the execution of $\mathsf{ObliFetch}$, and then randomly select $n_d$ $0s$ in the sequence to replace them with $1s$.
            Similarly, $\mathsf{Sim}^{\mathsf{ObliFetch}}_{CS_t}$ can also simulate the sequence of $\Phi_2$ generated during the execution of $\mathsf{ObliFetch}$.
            %
            %
            %
        \end{itemize}
        The proof of Theorem \ref{theo:secure} is completed.
    \end{proof}

    We now explicitly analyze how {\main} protects the data patterns that include
    dominance relationships among database tuples, the number of database tuples that each skyline tuple dominates, and the search access patterns.
    \begin{itemize}
        \item \textbf{Protecting the dominance relationships among data tuples.}
        Note that cloud servers initially invoke $\mathsf{ObliShuff}$ to obliviously shuffle the encrypted tuples in the original encrypted database $[\![{\mathcal{T}}]\!]^A$, and produce the shuffled database $\llb\widetilde{\mathcal{T}}\rrb^A$.
        The subsequent user-defined skyline finding process runs on $\llb\widetilde{\mathcal{T}}\rrb^A$.
        Therefore, even if {\csa} learn the dominance relationships between the tuples in the shuffled database $\widetilde{\mathcal{T}}$, they cannot infer the true dominance relationships among the tuples in the original database $\mathcal{T}$.

        \item \textbf{Protecting the number of database tuples that each skyline tuple dominates.}
        The execution of the BNL algorithm has the potential to reveal the number of database tuples dominated by each skyline tuple, which is a concern for privacy.
        However, in our oblivious realization of the BNL algorithm, as seen in lines 8-10 of Algorithm \ref{alg:fetch}, {\main} has a subtle design to address this issue.
        Specifically, when determining whether an encrypted database tuple $\llb\widetilde{\mathbf{t}}_i\rrb^A$ is dominated by an encrypted candidate skyline tuple $\llb\widetilde{\mathbf{t}}\rrb^A$, {\main} lets {\csa} open a masked version of their encrypted dominance flag $\llb\Phi_1\rrb^B$, rather than directly opening $\Phi_1$, as directly opening $\Phi_1$ has the potential to leak the number of database tuples dominated by each skyline tuple.
        With this strategy, {\csa} are not able to learn the correct dominance relationship between $\llb\widetilde{\mathbf{t}}_i\rrb^A$ and $\llb\widetilde{\mathbf{t}}\rrb^A$, so they are prevented from learning the number of database tuples that each skyline tuple dominates.

        \item \textbf{Protecting the search pattern.} For a user-defined skyline query $\llb Q'\rrb$, the cloud servers {\csa} only receive the secret shares $\langle Q'\rangle_{1}$ and $\langle Q'\rangle_{2}$, respectively. 
        According to the security of ASS \cite{demmler2015aby}, the secret shares are randomly generated and indistinguishable from uniformly random values.
        This means that even if the same private value is encrypted multiple times under ASS, the resulting ciphertext will still be different.
        Therefore, {\main} can hide the search pattern.

        \item \textbf{Protecting the access pattern.}
        %
        %
        Recall that in {\main} the skyline tuples are fetched from the encrypted sub-database $\llb\mathcal{C}\rrb^A$ by $\mathsf{ObliFetch}$.
        However, $\llb\mathcal{C}\rrb^A$ is a sub-database of the shuffled database $\llb\widetilde{\mathcal{T}}\rrb^A$.
        According to the security of oblivious shuffle, {\csa} cannot determine which tuples in the original database $T$ are skyline tuples, even if they learn which encrypted tuples in $\llb\mathcal{C}\rrb^A$ are skyline tuples (in the shuffled space).
        Therefore, {\csa} can hide the access pattern.

    \end{itemize}

    \section{Experiments}
    \label{sec:Experiments}

    \subsection{Setup}

    We implement {\main}'s protocols in C++.
    All experiments are performed on a workstation with 8 AMD Ryzen 7 5800H CPU cores and 16 GB RAM running 64-bit Windows 10.
    We run two threads in parallel on the same machine to simulate the two cloud servers.
    Besides, the network delay of the communication between the two cloud servers is set to 1 ms.
    Similar to the state-of-the-art prior work \cite{zhang2022efficient}, our empirical evaluations use a synthetic dataset which has 20 dimensions and 10000 tuples.
    It is noted that different user-defined skyline queries will give different constrained regions.
    This influences the size of the (encrypted) sub-database over which oblivious user-defined skyline fetching is actually performed.
    Therefore, we introduce the parameter \textit{selectivity} to control the sub-database size. 
    The selectivity is defined as the percentage of tuples that fall within the constrained region and is set to $0.1\%$ in our experiments unless otherwise stated.
    All results in our experiments are the average over 100 skyline queries.

    \subsection{Evaluation on Query Latency}
    
     \begin{figure*}[!t]
    	\centering
    	\begin{minipage}[t]{0.24\linewidth}
    		\centering{\includegraphics[width=\linewidth]{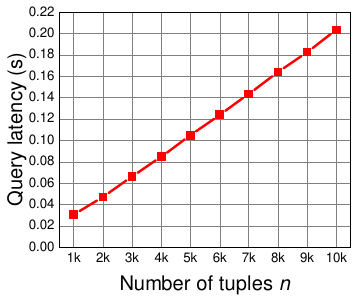}}
    		\caption{Query latency for varying number of tuples $n$ (with the number of dimensions $m=5$ and the number of selected dimensions $k=3$).}
    		\label{fig:ParaN}
    	\end{minipage}
    	\begin{minipage}[t]{0.24\linewidth}
    		\centering{\includegraphics[width=\linewidth]{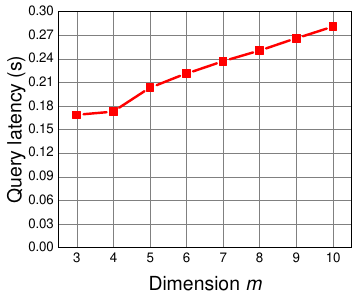}}
    		\caption{Query latency for varying number of dimensions $m$ (with the number of tuples $n=10000$ and the number of selected dimensions $k=3$).}
    		\label{fig:ParaM}
    	\end{minipage}
    	\begin{minipage}[t]{0.24\linewidth}
    		\centering{\includegraphics[width=\linewidth]{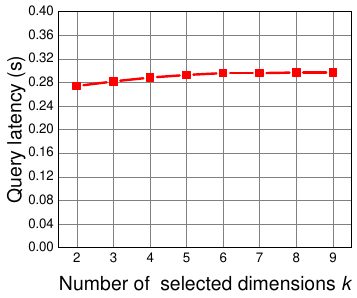}}
    		\caption{Query latency for varying number of selected dimensions $k$ (with the number of tuples $n=10000$ and the number of dimensions $m=10$).}
    		\label{fig:ParaK}
    	\end{minipage}
    	\begin{minipage}[t]{0.24\linewidth}
    		\centering{\includegraphics[width=\linewidth]{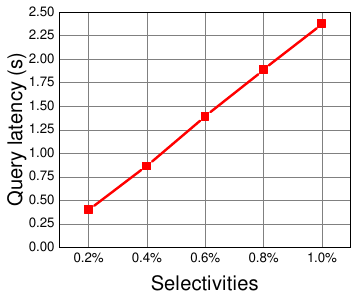}}
    		\caption{Query latency for varying selectivity (with the number of tuples $n=10000$, the number of dimensions $m=5$, and the number of selected dimensions $m=10$).}
    		\label{fig:ParaS}
    	\end{minipage}
        \vspace{-15pt}
    \end{figure*}

    \begin{figure*}[!t]
    	\centering
    	\begin{minipage}[t]{0.24\linewidth}
    		\centering{\includegraphics[width=\linewidth]{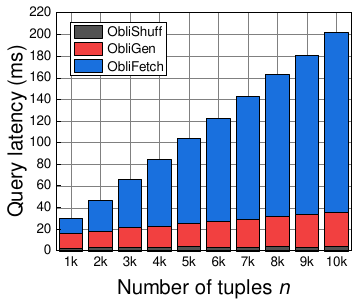}}
    		\caption{Breakdown of query latency for varying number of tuples $n$ (with $m=5,k=3$).}
    		\label{fig:FigN}
    	\end{minipage}
    	\begin{minipage}[t]{0.24\linewidth}
    		\centering{\includegraphics[width=\linewidth]{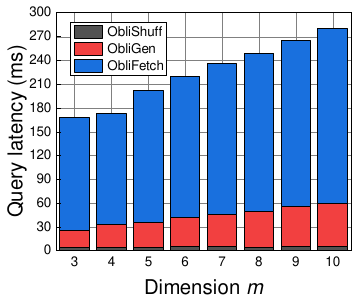}}
    		\caption{Breakdown of query latency for varying number of dimensions $m$ (with $n=10000,k=3$).}
    		\label{fig:FigM}
    	\end{minipage}
    	\begin{minipage}[t]{0.24\linewidth}
    		\centering{\includegraphics[width=\linewidth]{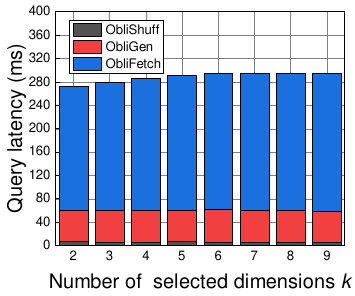}}
    		\caption{Breakdown of query latency for varying number of selected dimensions $k$ (with $n=10000,m=10$).}
    		\label{fig:FigK}
    	\end{minipage}
    	\begin{minipage}[t]{0.24\linewidth}
    		\centering{\includegraphics[width=\linewidth]{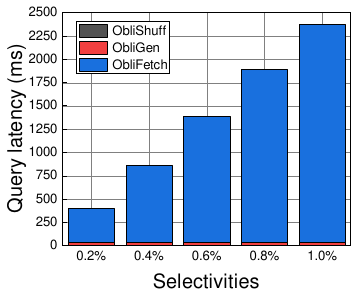}}
    		\caption{Breakdown of query latency for varying selectivity (with $n=10000,m=5,k=3$).}
    		\label{fig:FigS}
    	\end{minipage}
        \vspace{-15pt}
    \end{figure*}

    \begin{figure*}[!t]
        \centering
        \begin{minipage}[t]{0.24\linewidth}
            \centering{\includegraphics[width=\linewidth]{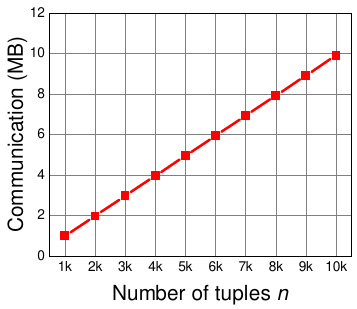}}
            \caption{Communication cost for varying number of tuples $n$ (with $m=5,k=3$).}
            \label{fig:CommN}
        \end{minipage}
        \begin{minipage}[t]{0.24\linewidth}
            \centering{\includegraphics[width=\linewidth]{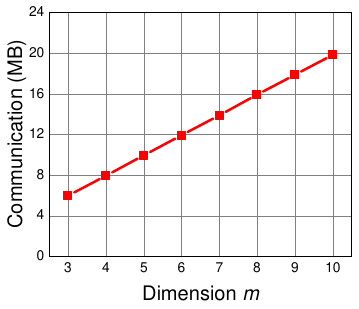}}
            \caption{Communication cost for varying dimensions $m$ (with $n=10000,k=3$).}
            \label{fig:CommM}
        \end{minipage}
        \begin{minipage}[t]{0.24\linewidth}
            \centering{\includegraphics[width=\linewidth]{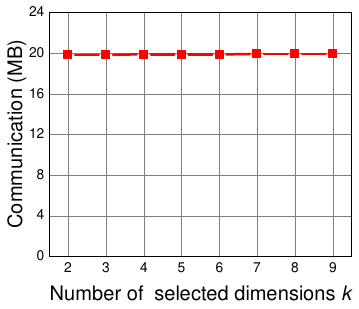}}
            \caption{Communication cost for varying number of selected dimensions $k$ (with $n=10000,m=10$).}
            \label{fig:CommK}
        \end{minipage}
        \begin{minipage}[t]{0.24\linewidth}
            \centering{\includegraphics[width=\linewidth]{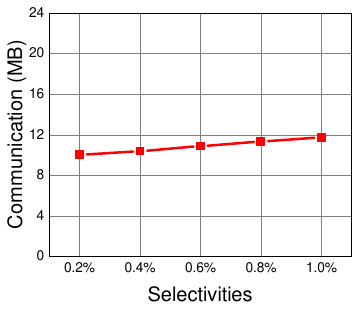}}
            \caption{Communication cost for varying selectivity (with $n=10000,m=5$, and $k=3$).}
            \label{fig:CommS}
        \end{minipage}
        \vspace{-15pt}
    \end{figure*}

    We first evaluate the query latency performance of {\main}.
    This refers to the time it takes the {\csa} to produce the result set of encrypted skyline tuples given the input of an encrypted user-defined skyline query.
    %
    We present the results in Fig. \ref{fig:ParaN}-Fig. \ref{fig:ParaS}, where we respectively vary the number of tuples $n\in\{1\times10^3,2\times10^3,3\times10^3,4\times10^3,5\times10^3,6\times10^3,7\times10^3,8\times10^3,9\times10^3,10\times10^3\}$, the number of dimensions $m\in\{3,4,5,6,7,8,9,10\}$, the number of selected dimensions $k\in\{2,3,4,5,6,7,8,9\}$, and the selectivity $\in\{0.2\%,0.4\%,0.6\%,0.8\%,1.0\%\}$.
    To illustrate the performance of {\main} more clearly, we present a concrete breakdown accordingly in Fig. \ref{fig:FigN}-Fig. \ref{fig:FigS}. 
    %
    It can be observed that the query latency increases linearly with the number of tuples (i.e., $n$), dimensions (i.e., $m$), and selectivity.
    However, from Fig. \ref{fig:ParaK}, we can observe that the query latency increases slowly with the increase in the number of selected dimensions (i.e., $k$).
    The reason is that in order to hide the target dimensions, {\main} pads each user-defined skyline query to the same number of dimensions as database tuples (i.e., $m$).
    Therefore, irrespective of the number of target dimensions, the cloud servers have to process each encrypted user-defined query with $m$ dimensions.
	However, increasing the number of target dimensions (i.e., $k$) leads to more complex constraints on the dominance relationship and results in more skyline tuples, thereby requiring more rounds of skyline finding. 
	As a consequence, the secure component $\mathsf{ObliFetch}$ consumes more time (as shown in Fig. \ref{fig:FigK}), causing the query latency to increase slowly with the increase in $k$.

    %
    %
    %

    \subsection{Evaluation on Communication Performance}

    We now evaluate the online communication performance of {\main}.
    The same parameter settings as the above are used.
    The results are depicted in Fig. \ref{fig:CommN}, Fig. \ref{fig:CommM}, Fig. \ref{fig:CommK} and Fig. \ref{fig:CommS}.
    It is worth noting that the total communication cost increases linearly with $n$, $m$ and the selectivity, but remains almost constant when $k$ is altered.
    For instance, as $n$ increases from $1000$ to $10000$, the communication cost with $m=5$ and $k=3$ increases from approximately $1$MB to $10$MB, but as $k$ increases from $2$ to $9$, the communication cost with $n=10000$ and $m=10$ remains approximately $20$MB.

    \subsection{Comparison to State-of-the-art Prior Work}
   
    We emphasize that a fair performance comparison between {\main} and the state-of-the-art prior work \cite{zhang2022efficient} does not exist due to the security downsides of the latter, as analyzed in Section \ref{sec:Related}.
    Specifically, the work \cite{zhang2022efficient} fails to protect the dominance relationships among the database tuples, the number of database tuples that each skyline tuple dominates, and the access pattern (which in turn leaks the search pattern).
    In contrast, {\main} allows the cloud to obliviously process user-defined skyline queries while keeping all these factors private.
    Despite achieving stronger security guarantees, we note that {\main} still outperforms \cite{zhang2022efficient} in some aspects. 
    For example, when the number of tuples $n$ increases from $1000$ to $10000$, the time for database encryption in \cite{zhang2022efficient} increases from approximately $8$s to $150$s. 
    Similarly, when the number of dimensions $m$ increases from $3$ to $10$ for $n=10000$, the time for query token generation in \cite{zhang2022efficient} increases from approximately $10$ms to $120$ms\footnote{The results are reported in \cite{zhang2022efficient}. As the authors do not give a specific number of dimensions (i.e., $m$), here we assume that $m$ is the maximum value in their experiments, i.e., $m=10$.}.
     In contrast, {\main} can encrypt the database or the user-defined skyline query in just 1ms for the same database size. 
     %
    Although the work \cite{zhang2022efficient} may be more efficient than {\main} in certain scenarios, we emphasize that its efficiency sacrifices security, which hinders its practical usability.
    For instance, under the dataset with $n=10000$, $m=5$, and $k=3$, while the query latency of the approach presented in \cite{zhang2022efficient} is about $18$ms for a selectivity of $0.1\%$ and $90$ms for a selectivity of $1.0\%$, the query latency of {\main} is about $0.2$s and $2.4$s respectively (which should also not have significant perceivable difference in user experience). Yet {\main}'s superior security guarantees outweigh the efficiency gains of \cite{zhang2022efficient}.

    \section{Conclusion}   
    \label{sec:Conclusion}
    
    In this paper, we devise, implement and evaluate {\main}, a new framework for oblivious user-defined skyline queries over encrypted outsourced database with stronger security guarantees over prior art.
    {\main} is designed to not only offer confidentiality protection for the outsourced database, user-defined skyline query, and query results, but also ambitiously for data patterns, including dominance relationships among the database tuples, the number of database tuples that each skyline tuple dominates, and search access patterns.
    Through extensive experiments, we show that while achieving stronger security guarantees over the state-of-the-art, {\main} exhibits superior database and query encryption efficiency as well as practical query latency performance.

\balance

\bibliographystyle{IEEEtran}
\bibliography{my}

\end{document}